\documentclass[copyright,creativecommons]{eptcs}
\providecommand{\event}{FICS 2024} 

\usepackage{multirow}
\usepackage{hyperref}
\usepackage{xspace}
\usepackage{amsmath,amssymb}
\usepackage{stmaryrd} 

\usepackage{apxproof} 
\usepackage{algorithm}
\usepackage{algpseudocode}

\usepackage{todonotes} 

\usepackage{tikz-cd}

\usepackage{tikz}
 \usetikzlibrary{trees}
 \usetikzlibrary{shapes}
 \usetikzlibrary{fit}
 \usetikzlibrary{shadows}
 \usetikzlibrary{backgrounds}
 \usetikzlibrary{arrows,automata}

\tikzset{
   n/.style= {circle,fill,inner sep=1.5pt,node distance=2cm}
  ,acc/.style={circle,draw,inner sep=3pt,node distance=2cm}
  ,phantom/.style={circle},
  ,arr/.style={->, >=stealth, semithick, shorten <= 3pt, shorten >= 3pt}
}

\newcommand{\sem}[1]{\llbracket #1 \rrbracket}

\newcommand{\takeout}[1]{\empty}

\theoremstyle{plain}

\newtheorem{thm}{Theorem}

\newtheorem{lemma}[thm]{Lemma}

\theoremstyle{definition}

\newtheorem{definition}[thm]{Definition}

\title{Faster Game Solving by Fixpoint Acceleration}

\author{Daniel Hausmann\footnote{Supported by  the ERC Consolidator grant D-SynMA (No.
	772459)}
\institute{University of Liverpool, United Kingdom}
\email{hausmann@liverpool.ac.uk}
}
\def\titlerunning{Faster Game Solving by Fixpoint Acceleration}
\def\authorrunning{D. Hausmann}

\begin{document}

\maketitle
\begin{abstract}
We propose a method for solving parity games with acyclic (DAG) sub-structures by computing nested fixpoints of a DAG attractor function that lives over the non-DAG parts of the game, thereby restricting the domain of the involved fixpoint operators. Intuitively, this corresponds to accelerating fixpoint computation by inlining cycle-free parts during the solution of parity games, leading to earlier convergence. We also present an economic later-appearance-record construction that takes Emerson-Lei games to parity games, and show that it preserves DAG sub-structures; it follows that the proposed method can be used also for the accelerated solution of Emerson-Lei games.
\end{abstract}

\section{Background}

The analysis of infinite duration games is of central importance to
various problems in theoretical computer science such as formal verification (model checking), logical reasoning (satisfiability checking), or automated program construction (reactive synthesis). Previous work has shown
how fixpoint expressions can be used to characterize winning in such games,
which in turn has enabled the development of symbolic game solving algorithms that circumvent the state-space explosion to some extent and therefore
perform reasonably well in practice in spite of the high complexity of the
considered problems.

In this work, we build on the close connection between game solving and
fixpoint computation to obtain a method that accelerates the solution of
parity games with acyclic sub-structures. Intuitively, we base our method
on the observation that cycle-free parts in parity games can be dealt with
summarily by using the computation of attractors in place of handling
each node individually. In games that are largely cycle-free, but in which
attraction along cycle-free parts can be evaluated without exploring most of
their nodes, this significantly speeds up the game solution process.

The close relation between games and fixpoint expression has been well
researched. On one hand, it has been shown that winning regions in games with various objectives can be characterized by fixpoint expressions ~\cite{BruseFL14,BloemJPPS12,PitermanP06,HausmannEA23}. The result for parity games (e.g.~\cite{BruseFL14}) arguably is the best-known case: it has been shown
that the winning region in parity games with $k$ priorities can be specified
by a modal $\mu$-calculus formula (sometimes referred to as \emph{Walukiewicz formulas}) with alternation-depth $k$; this result connects parity game solving and model checking for the $\mu$-calculus. It corresponds to arrow (1) in the diagram below, which is meant to illustrate the context of our contribution.
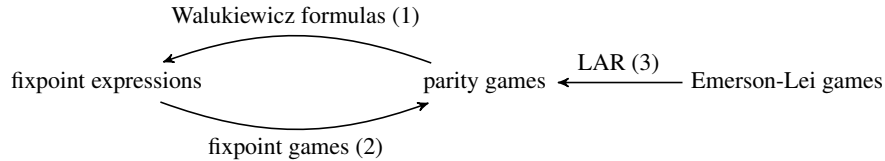
\begin{center}

\begin{footnotesize}
\begin{tikzpicture}[node distance=1.0cm,>=stealth',auto,semithick,        shorten > = 1pt]

  \tikzstyle{place}=[circle,thick,draw=blue!75,fill=blue!20,minimum size=6mm]
  \tikzstyle{red place}=[place,draw=red!75,fill=red!20]
  \tikzstyle{transition}=[rectangle,thick,draw=black!75,
  			  fill=black!20,minimum size=4mm]

    \tikzstyle{every state}=[
        draw = black,
        thick,
        fill = white,
        minimum size = 4mm
    ]

    \node (1) at (0,1) {fixpoint expressions};
    \node (3) at (5,1)
    {parity games};
    \node (5) at (9,1)
    {Emerson-Lei games};

    \path[->] (1) edge  [bend right=20] node[below]  {fixpoint games (2)} (3);
    \path[->] (3) edge  [bend right=20] node[above]  {Walukiewicz formulas (1)} (1);
    \path[->] (5) edge  node[above]  {LAR (3)}  (3);

\end{tikzpicture}
\end{footnotesize}
\end{center}

A fruitful reduction in the converse direction is more recent. \emph{Fixpoint games}, that is, parity games of a certain structure (defined in detail below), have
been used to characterize the semantics of hierarchical fixpoint equation systems in terms of games~\cite{BaldanKMP19}. While this reduction (corresponding to arrow (2) in the diagram) incurs blow-up that is exponential in 
the size of the domain of the equation system, it still proves helpful as it allows
for reasoning about nested least and greatest fixpoints in terms of (strategies in) parity games.
Fixpoint games have also been used to lift the recent breakthrough 
quasipolynomial result for parity game solving to
the evaluation of general fixpoint equation systems~\cite{HausmannSchroeder21}.

Later-appearance-record (LAR) constructions have been used to reduce games with richer winning objectives to parity games; this subsumes in particular the LAR-reduction for Emerson-Lei games (arrow (3)), which have recently garnered attention due to their succinctness and favorable closure properties.

In this context, the contribution of the current work is two-fold. First, we show
how for parity games with cycle-free sub-structures, Walukiewicz' construction 
can be adapted to use multi-step attraction along cycle-free parts (we call such parts \emph{DAGs in games}) in place of
one-step attraction. In consequence, the domain of the resulting fixpoint
expression can be restricted to the parts of the game that are not cycle-free,
thereby accelerating convergence in the solution of the fixpoint expression.
If multi-step attraction can be evaluated without visiting most cycle-free nodes
(which is the case, e.g., when cycle-free parts encode predicates that can be efficiently evaluated), this trick has been shown to significantly speed up game solution.
Applications can be found in both model checking~\cite{HausmannS19} and satisfiability checking~\cite{HausmannS19a}
for generalized $\mu$-calculi in which satisfaction and joint satisfiability,
respectively,
of the modalities can be harder to verify than in the standard $\mu$-calculus.

In addition to that, we propose a later-appearance-record construction for the
transformation of Emerson-Lei games to equivalent parity games (corresponding to arrow (3) in the above diagram) and show that the
reduction preserves cycle-free sub-structures, thereby enabling usage
of the proposed acceleration method also for games with full Emerson-Lei objectives. 

\section{Games and Fixpoint Expressions}\label{sec:gamesfp}

We start by introducing notions of games and fixpoint expressions
pertaining to the acceleration method proposed in Sections~\ref{sec:dag} and~\ref{sec:solve} below.

\emph{Games.}
An \emph{arena} is a graph $A=(V_\exists,V_\forall,E)$, consisting of a set $V=V_\exists\cup V_\forall$ of \emph{nodes} and a set $E\subseteq V\times V$ of \emph{moves}; furthermore,
the set of nodes is partitioned
into the sets $V_\exists$ and $V_\forall$
of nodes \emph{owned} by player $\exists$ and
by player $\forall$, respectively.
We write $E(v)=\{w \in V\mid (v,w)\in E\}$ for the set of nodes reachable from
node $v\in V$ by a single move. We assume without loss of generality that
every node has at least one outgoing edge, that is, that $E(v)\neq\emptyset$
for all $v\in V$.
A \emph{play} $\pi=v_0 v_1\ldots$
on $A$ is a (finite or infinite) sequence of nodes such that
$v_{i+1}\in E(v_i)$ for all $i\geq 0$. 
By abuse of notation we denote by $A^\omega$ the set of infinite plays
on $A$ and by $A^+$ the set of finite (nonempty) plays on $A$.
A \emph{strategy} for player $i\in\{\exists,\forall\}$
is a function $\sigma:V^*\cdot V_i\to V$ that assigns a node
$\sigma(\pi v)\in V$ to every finite play $\pi v$ that ends in a node
$v\in V_i$. A strategy $\sigma$ for player $i$ is said to be \emph{positional} if the prescribed moves
depend only on the last nodes in plays; formally, this is the case if we have $\sigma(\pi v)=\sigma(\pi' v)$
for all $\pi,\pi'\in V^*$ and $v\in V_i$.
A play $\pi=v_0 v_1\ldots$ then is \emph{compatible} with a strategy $\sigma$ for
player $i\in\{\exists,\forall\}$ if
for all $j\geq 0$ such that $v_j\in V_i$, we have $v_{j+1}=\sigma(v_0 v_1\ldots v_j)$.
A strategy for player
$i\in\{\exists,\forall\}$ \emph{with memory $M$} is a tuple
$\sigma=(M,m_0,\mathsf{update}:M\times E\to M,\mathsf{move}:V_i\times M\to V)$,
where $M$ is some set of \emph{memory values}, $m_0\in M$ is the initial memory value,
the \emph{update function} $\mathsf{update}$ assigns the outcome $\mathsf{update}(m,e)\in M$ of updating the
memory value $m\in M$ according to the effects of taking the move $e\in E$, and
the moving function $\mathsf{move}$ prescribes a single move $(v,\mathsf{move}(v,m))\in E$ to
every game node $v\in V_i$ that is owned by player $i$, depending on the memory value $m$.
We extend $\mathsf{update}$ to finite plays $\pi$ by putting
$\mathsf{update}(m,\pi)=m$ in the base case that $\pi$ consists of a
single node,
and by putting $\mathsf{update}(m,\pi)=\mathsf{update}(\mathsf{update}(m,\tau),(v,w))$
if $\pi$ is of the shape $\tau v w$, that is, contains at least two nodes.
Then we obtain a strategy (without explicit memory) $\tau_\sigma:V^*\cdot V_i\to V$ from a strategy $\sigma=(M,m_0,\mathsf{update}:M\times E\to M,\mathsf{move}:V_i\times M\to V)$ with memory by putting, for all
$v_0 v_1\ldots v_i\in V^*\cdot V_i$,
\begin{align*}
\tau_\sigma(v_0 v_1\ldots v_i)=\mathsf{move}(v_i,\mathsf{update}(m_0,v_0 v_1\ldots v_i)).
\end{align*}

In this work we consider two types of objectives: 
\emph{parity objectives} and, more generally, \emph{Emerson-Lei objectives}.

Emerson-Lei objectives~\cite{HunterD05} are
specified relative to a coloring function $\gamma_C:V\to 2^C$ (for some set
$C$ of colors) that assigns a set
$\gamma_C(v)\subseteq C$ of colors to each node $v\in V$. A play $\pi=v_0v_1\ldots$ then induces a
sequence
$\gamma_C(\pi)=\gamma_C(v_0)\gamma_C(v_1)\ldots$ of sets of colors.
Emerson-Lei objectives are given as Boolean formulas
$\varphi_C\in\mathbb{B}(\{\mathsf{Inf}(c)\mid c\in C\})$
over atoms of the shape $\mathsf{Inf}(c)$; throughout, we write
$\mathsf{Fin}(c)$ for $\neg (\mathsf{Inf}(c))$.
Such formulas are interpreted over infinite sequences
$\gamma_0\gamma_1\ldots$ of sets of colors. We
put $\gamma_0\gamma_1\ldots\models \mathsf{Inf}~c$ if and only if there are infinitely many positions $i$ such that $c\in\gamma_i$; satisfaction of Boolean operators is defined in the usual way.
Then a play $\pi$ on $A$ \emph{satisfies} the objective $\varphi_C$ if and only if
$\gamma_C(\pi)\models \varphi_C$ and we define the Emerson-Lei objective induced by $\gamma_C$
and $\varphi_C$ by putting
\begin{align*}
\alpha_{\gamma_C,\varphi_C}=\{\pi\in V^\omega\mid\gamma_C(\pi)\models\varphi_C\}.
\end{align*}

An \emph{Emerson-Lei game} is a tuple $G=(A,\alpha)$, where
$A=(V_\exists,V_\forall,E)$ is an arena and $\alpha$ is an Emerson-Lei objective.
A strategy $\sigma$ for player $\exists$ is \emph{winning} at some node $v\in V$
if all infinite plays that start at $v$ and are compatible with $\sigma$ satisfy the
objective $\alpha$.
A strategy $\tau$ for player $\forall$ is defined winning dually.
The \emph{winning region} $\mathsf{Win}_G$ of player $\exists$ in the game $G$ is the set of nodes for
which there is a winning strategy for player $\exists$.

As a special case, we define \emph{parity games} to be Emerson-Lei games $(A,\alpha)$ where $\alpha$ is the Emerson-Lei objective induced by a coloring function $\gamma:V\to 2^{\{p_1,\ldots,p_k\}}$
that assigns exactly one of the colors $\{p_1,\ldots,p_k\}$ to each game node,
and the formula
\begin{align*}
\textstyle\bigvee_{i \text{ even}}\mathsf{Inf}(p_i)\wedge \textstyle\bigwedge_{j>i} \mathsf{Fin}(p_j),
\end{align*}
expressing that for some even $i$, color $p_i$ is visited infinitely often
and no color $p_j$ such that $j$ is larger than $i$ is visited infinitely often.
We denote such games by $(V_\exists,V_\forall,E,\Omega:V\to\mathbb{N})$,
where $\Omega$ is a \emph{priority function}, assigning to each node $v\in V$
the number $\Omega(v)=i$ such that $\gamma(v)=\{p_i\}$.

Emerson-Lei games are known to be \emph{determined}, that is, every node in such games
is won by exactly one of the two players.
\emph{Solving} a game then amounts to computing the winner for each game node.
While the solution of Emerson-Lei games is known to be a \textsc{PSpace}-complete
problem~\cite{HunterD05}, the precise complexity of solving parity games
remains an open question; it is known however, that solving parity games is
in \textsc{NP}$\cap$\textsc{Co-NP} as well as in \textsc{QP} (that is,
can be done deterministically in quasipolynomial time).
Winning strategies in Emerson-Lei games may require exponential memory (specifically, 
$k!$, where $k$ is the number of colors), but parity games have positional strategies~\cite{DziembowskiJW97}.\medskip

\emph{Fixpoint Expressions.} Let $U$
be a set, $k$ a natural number
(we assume without loss of generality that $k$ is even) and let $f:\mathcal{P}(U)^{k}\to \mathcal{P}(U)$ be
a monotone function (such that $f(W_1,\ldots,W_k)\subseteq f(W'_1,\ldots,W'_k)$ whenever
$W_i\subseteq W'_i$ for all $1\leq i\leq k$).
This data induces a fixpoint expression
\begin{align*}
e=\nu X_{k}.\,\mu X_{k-1}.\,\ldots\,.\,\mu X_1.\, f(X_1,\ldots, X_k).
\end{align*}
The semantics of such fixpoint expressions is defined as usual, following
the Knaster-Tarski fixpoint theorem.
Formally, we put
\begin{align*}
\mathsf{LFP}\, g &= \bigcap\{W\subseteq U\mid g(W)\subseteq W\} &
\mathsf{GFP}\, g &= \bigcup\{W\subseteq U\mid W\subseteq G(W)\}
\end{align*}
for monotone functions $g:\mathcal{P}(U)\to \mathcal{P}(U)$ and define
\begin{align*}
\sem{e}_i(X_{i+1},\ldots X_{k})=
\begin{cases}
\mathsf{LFP}(\lambda X_i. \sem{e}_{i-1}(X_i,X_{i+1},\ldots X_{k})) & \text{ if $i$ is odd} \\
\mathsf{GFP}(\lambda X_i. \sem{e}_{i-1}(X_i,X_{i+1},\ldots X_{k})) & \text{ if $i$ is even}
\end{cases}
\end{align*}
for $1\leq i\leq k$ and $X_{i+1},\ldots, X_k\subseteq V$, where $\sem{e}_{0}(X_1,\ldots X_{k})=f(X_1,\ldots X_{k})$.
Then we say that $v\in V$ is \emph{contained} in $e$ (denoted $v\in e$)
if and only
if $v\in\sem{e}_k$; by abuse of notation, we denote
$\sem{e}_k$ just by $e$.\medskip

\emph{Relation between Games and Fixpoint Expressions.}
The winning region in parity games can be specified
by a fixpoint expression over the underlying game arena.
To see how this works, let $G=(V_\exists, V_\forall,E,\Omega)$ be a parity game
and define the following notation.
Let $\Omega^{-1}(i)=\{v\in V\mid \Omega(v)=i\}$ denote the set
of nodes that have priority $i$.
The \emph{controllable predecessor function}
$\mathsf{Cpre}:\mathcal{P}(V)\to \mathcal{P}(V)$ that computes one-step attraction for player $\exists$ in $G$
is defined, for $X\subseteq V$, by putting
\begin{align*}
\mathsf{Cpre}(X)=\{v\in V_\exists\mid E(v)\cap X\neq\emptyset\}\cup 
\{v\in V_\forall\mid E(v)\subset X\}.
\end{align*}
Thus $\mathsf{Cpre}(X)$ contains all nodes $v\in V_\exists$ for which
some outgoing move leads to $X$, as well as all nodes $v\in V_\forall$ for which
all outgoing moves lead to $X$; intuitively, this is the set of nodes from which player $\exists$ can force that $X$ is reached in one step of the game.
The characterization is given by the fixpoint expression
\begin{equation}
\mathsf{parity}_G=\nu X_{k}.\,\mu X_{k-1}.\,\ldots\,.\,\mu X_1.\, 
\bigvee_{1\leq i\leq k} \Omega^{-1}(i)\wedge\mathsf{Cpre}(X_i).
\end{equation}
The characterization result (e.g.~\cite{BruseFL14}) then is stated as follows.\medskip

\begin{lemma}\label{lem:parityfp}
We have $\mathsf{Win}_G=\mathsf{parity}_G$.
\end{lemma}
\medskip
Similar fixpoint characterizations of winning regions have been given for games with more involved
objectives, including GR[1] conditions~\cite{BloemJPPS12}, Rabin and Streett objectives~\cite{PitermanP06}, and most recently, also for full Emerson-Lei objectives~\cite{HausmannEA23}. While the characterizations of winning according to these more general objectives involve fixpoint expressions of a more general format (given by hierarchical systems of fixpoint equations)
than the one we have defined above,
the fixpoint for the winning region in parity games will be sufficient for
the purposes of this work.
All mentioned fixpoint characterizations of winning regions allow for (symbolic) solution of games by computing nested fixpoints,
using for instance standard fixpoint iteration~\cite{Seidl96}, or using more recent techniques that employ
universal trees to obtain quasipolynomial algorithms to compute nested fixpoints~\cite{HausmannSchroeder21,ArnoldEA21}.\medskip

For the converse direction, that is, going from fixpoints to games,
it has recently been shown (in a more general framework than the one used in this work, see~\cite{BaldanKMP19} for details) that the semantics of fixpoint expressions
can be equivalently given in terms of parity games. The \emph{fixpoint game}
that is associated with a fixpoint expression 
\begin{align*}
e=\nu X_{k}.\,\mu X_{k-1}.\,\ldots\,.\,\mu X_1.\, f(X_1,\ldots, X_k)
\end{align*} over $U$ is the parity game $G_e=(V_\exists,V_\forall,E,\Omega:V\to\{1,\ldots, k\})$ played over the sets
$V_\exists=U$ and $V_\forall=\mathcal{P}(U)^{k}\cup (U\times[k])$ of nodes; 
moves and priorities are given by the following table,
abbreviating $(U_1,\ldots, U_k)$ by $\overline{U}$.
\begin{center}
\begin{tabular}{|c|c|c|c|}
\hline
Node & owner & priority & allowed moves \\
\hline
\hline
$u\in U$ & $\exists$ & $0$ & $\{\overline{U}\in\mathcal{P}(U)^k\mid u\in f(\overline{U})\}$\\
$\overline{U}\in \mathcal{P}(U)^k$ & $\forall$ & $0$ & $\{(u,i)\in U\times[k]\mid u\in V_i\}$\\
$(u,i)\in U\times[k] $ & $\forall$ & $i$ & $\{u\}$\\
\hline
\end{tabular}
\end{center}
When the game reaches a node $u\in U$, player $\exists$ thus has
to provide a tuple $\overline{U}=(U_1,\ldots, U_k)\in\mathcal{P}(U)^k$ such that
$u\in f(\overline{U})$. Player $\forall$ in turn can challenge the provided
tuple by moving to any $(u',i)$ such that $u'\in U_i$, and then continuing the game
at the node $u'\in U$. Crucially, such a sequence triggers the priority $i$ which
is an even number if $X_i$ belongs to a greatest fixpoint in $e$, and odd otherwise. We observe that the number of nodes in $G_e$ is exponential in $|U|$; however, player $\exists$ owns only $|U|$ of these nodes.

\begin{lemma}[\cite{BaldanKMP19}]
We have $e=\mathsf{Win}_{G_e}$.
\end{lemma}
While fixpoint games enable reasoning about fixpoint expressions by
reasoning about parity games, the exponential size of fixpoint games makes it unfeasible to evaluate fixpoint expressions
by solving the associated fixpoint games.

\section{Game Arenas with DAG Sub-structures}\label{sec:dag}

In this section we introduce notions related to cycle-free sub-structures in (parity) games that will be central to the
upcoming acceleration result.

\begin{definition}
Let $G=(V_\forall,V_\exists,E)$ be a game arena. A set
$W\subseteq V$ of nodes is a \emph{DAG} (directed acyclic graph) if it does not contain an $E$-cycle; then there is no play of $G$ that eventually stays within $W$ forever. 
A DAG is not necessarily connected, that is, it may consist of several subgames of $G$.
Given a DAG $W\subseteq V$, we write $V'=V\setminus W$ and refer
to the set $V'$ as \emph{real} nodes (with respect to $W$).
A DAG is \emph{positional} if for each existential player node $w\in W\cap V_\exists$ in it,
there is exactly one real node $v\in V'$ from which $w$ is reachable without visiting
any other real node.
\end{definition}

We will be interested in parity games over game arenas with DAG structures, and
in the computation of attraction (that is, alternating reachability) within such DAGs.
To this end, we introduce the following notation of DAG attraction in parity games.

\begin{definition}
Let $G=(V_\forall,V_\exists,E,\Omega)$ be a parity game with priorities
$\{1,\ldots,k\}$.
Given a DAG $W$ and $k$ sets $\overline{V}=(V_1,\ldots, V_k)$ of real nodes and a real node $v\in V'$, we say that player $\exists$ can \emph{attract} to $\overline{V}$ from $v$ within $W$
if there is a positional strategy $\sigma$ for player $\exists$ such that for all plays $\pi$ that start at $v$ and adhere to $\sigma$, the first node $v'$ in $\pi$ such that $v'\neq v$ and
$v'\in V'$ is contained in $V_p$, where $p$ is the maximal priority that is visited
by the part of $\pi$ that leads from $v$ to $v'$.

Given a DAG $W$, we define the \emph{DAG attractor function}
$\mathsf{DAttr}_W:\mathcal{P}(V')^k\to \mathcal{P}(V')$ by
\begin{equation}
\mathsf{DAttr}_W(\overline{V})=\{v\in V\mid
\text{player $\exists$ can attract to $\overline{V}$ from $v$ within $W$}\}\label{eq:dag}
\end{equation}
for $\overline{V}=(V_1,\ldots,V_k)\in \mathcal{P}(V')^k$.
We assume a bound $t_{\mathsf{Attr}_W}$ on the runtime of computing, 
for any input $\overline{V}\in\mathcal{P}(V')^k$, the dag attractor of $\overline{V}$ through $W$.

\end{definition}

We always have $t_{\mathsf{Attr}_W}\leq |E|$, using a least fixpoint computation to check for alternating reachability, thereby possibly exploring all DAG edges of the game. The method that we propose in the next section 
can play out its strength best in parity games that have large DAGs that can
be efficiently evaluated (more formally, this typically means that we have both $|V'|<\log|V|$ and $t_{\mathsf{Attr}_W}<\log |V|$).

\section{Large-step Solving for Parity Games}\label{sec:solve}

In this section we show that a parity game with DAG sub-structure $W$
can be solved by computing a fixpoint of the DAG attractor function, using
$V'=V\setminus W$ as domain of the fixpoint computation, rather than $V$. 

So let $G=(V_\exists,V_\forall,E,\Omega)$ be a parity game with (positional)
DAG $W$.
Recall (from Lemma~\ref{lem:parityfp}) that the winning region in $G$ is given by
\begin{align*}
\mathsf{Win}_{G}=\nu X_{k}.\,\mu X_{k-1}.\,\ldots\,.\,\mu X_1.\, 
\bigvee_{1\leq i\leq k} \Omega^{-1}(i)\wedge\mathsf{Cpre}(X_i),
\end{align*}
where $\mathsf{Cpre}$ operates on subsets $X_i$ of $V$. We show that
\begin{equation}
\mathsf{Win}_G\cap V'=\nu Y_{k}.\,\mu Y_{k-1}.\,\ldots\,.\,\mu Y_1.\,
\mathsf{DAttr}_W(Y_1,\ldots,Y_k)\label{eq:sfp},
\end{equation}
pointing out that  $\mathsf{DAttr}_W(Y_1,\ldots,Y_k)$
is a function over subsets $Y_1,\ldots Y_k$ of $V'$.
Therefore Equation (\ref{eq:sfp}) shows how attention can be reduced to non-DAG nodes (from $V'$) when solving $G$; however, the price for this is that the function
$\mathsf{DAttr}_W$
of which the fixpoint is computed is complex since it computes multi-step attraction within a DAG
in place of one-step attraction as encoded by $\mathsf{Cpre}$.

We make this intuition precise and prove the main result (that is, Equation (\ref{eq:sfp})) as follows.

\begin{lemma}\label{lem:pgtrick}
Let $G$ be a parity game with priorities $1$ to $k$ and set $V$ of nodes, let $W$ be a positional DAG in $G$, and
let $V'=V\setminus W$ and $m:=|V'|$. Then
$\mathsf{Win}_G\cap V'$ can be computed with $\mathcal{O}(m^{\log k})$ computations of a DAG attractor;
if $k<\log m$, then $\mathsf{Win}_G\cap V'$ can be computed with a number of DAG attractor computations
that is polynomial in $m$.
\end{lemma}

\begin{proof}
We show that the winning regions in $G$ can be computed by a fixpoint expression
(with alternation-depth $k$) of the DAG attractor function given in
Equation (\ref{eq:dag}).
Without loss of generality, assume that $k$ is an even number.

We define the fixpoint expression
\begin{align*}
d&=\nu X_{k}.\,\mu X_{k-1}.\,\ldots\,.\,\mu X_1.\,
\mathsf{DAttr}_W(X_1,\ldots,X_k),
\end{align*}
noting that the fixpoint variables $X_i$ in $d$ range over subsets of $V'$ rather than over 
subsets of $V$.
It has been shown in~\cite{HausmannSchroeder21} how universal trees can 
be used to compute the nested fixpoint $d$
with $\mathcal{O}(m^{\log k})$ (that is, quasipolynomially many) iterations of the function
$\mathsf{DAttr}_W$.
It follows from the results of the same paper that if $k<\log m$, then $d$
can be computed with a polynomial (in $m$) number of computations of a dag attractor.
It remains to show that $d$ is the winning region (restricted to $V'$)
of player $\exists$ in $G$.
\begin{itemize}
	\item[$\Rightarrow$]

     Let $s$ be a strategy with which player $\exists$ wins from all nodes in their
     winning region in $G_d$.
     We define a positional strategy $t$ for player $\exists$ in $G$. The definition
     proceeds in steps that lead from a real node to the next real nodes. 
     Given a real node $v\in V'$ that is contained in the winning region of player $\exists$
     in $G_d$, we have $s(v)=\overline{V}\in\mathcal{P}(V')^{k}$ such that
     player $\exists$ can attract to $\overline{V}$ from $v$ within $W$. For the dag
     part of $G$ that starts at $v$, define the strategy $t$ to simply use the positional
     strategy that attracts to $\overline{V}$ from $v$ within $W$.
	
     To see that $t$ is a winning strategy, let $\pi$ be a play of $G$ that is compatible
     with $t$. Then $\pi$ induces a play $\tau$ of $G_d$ that is compatible with 
     $s$: For any two positions $i$ and $i'$ such that $i<i'$,
     both $\pi(i)$ and $\pi(i')$ are real nodes and there is no $i<j<i'$ such that $\pi(j)$ is 
     a real node, let $p$ be the maximal priority that is visited by the partial play
     from $\pi(i)$ to $\pi(i')$.
     By construction we have $s(\pi(i))=(V_1,\ldots,V_p,\ldots,V_k)$ where $V_p$ contains
     $\pi(i')$. Thus the partial play of $G$ that leads from $\pi(i)$ to $\pi(i')$
     induces the partial play $\tau_i=\pi(i) s(\pi(i)) (\pi(i'),p) \pi(i')$ of $G_d$.
     We define $\tau$ to be the concatenation of all partial plays $\tau_i$.
     By construction, $\tau$ is compatible with $s$; as $s$ is a winning strategy,
     player $\exists$ wins $\tau$. Furthermore,
     the maximal priorities that are visited infinitely often in $\tau$ and $\pi$ agree.
     Thus player $\exists$ wins $\pi$, as required.
	
	\item[$\Leftarrow$]
	
     Let $t$ be a positional strategy for player $\exists$ in $G$ with which they win
     from all nodes in their winning region. We define a positional strategy $s$ for
     player $\exists$ in $G_d$ as follows. For
     each node $v\in V'$ that is contained in the winning region of player $\exists$
     in $G$, $t$ yields a positional strategy that attracts from $v$
     to some $\overline{V}=(V_0,\ldots,V_k)\in\mathcal{P}(V')^{k}$ within $W$:
     For all $v'$ and all $0\leq p\leq k$
     such that there is a partial play of $G$ that is compatible with $t$, starts at
     $v$, ends at $v'$, and in between visits only nodes from $W$, and in which the
     maximal priority that is visited is $p$, add $v'$ to $V_p$. Put $s(v)=\overline{V}$.
     As $v$ attracts from $v$ to $\overline{V}$ within $W$ by construction, $s$
     is a valid strategy.
     
     To see that $s$ is a winning strategy, let $\tau$ be a play of $G_d$
     that starts at a real node from the winning region of player $\exists$ in $G$ 
     and that is compatible with $s$. For each partial sub-play
     $v_i s(v_i) (v_{i+1},p_i) v_{i+1}$ of $\tau$, we define $\pi_i$
     to be some partial play of $G$ that is compatible
     with $t$, starts at $v_i$, ends at $v_{i+1}$, 
     and in between visits only nodes from $W$, and in which the
     maximal priority that is visited is $p_i$. Such a play exists by construction of $s$.
     Let $\pi$ denote the concatenation of all $\pi_i$.
     Again, the maximal priorities that are visited infinitely often in $\tau$ and $\pi$
     agree. As $t$ is a winning strategy, player $\exists$ wins both $\tau$ and $\pi$.

\end{itemize}
\end{proof}

In case that $m<\log n$ and $t_{\mathsf{Attr}_W}\in \mathcal{O}(\log n)$ (that is, DAG attractability can be decided without exploring most of the DAG nodes), this trick leads to exponentially faster game solving.\medskip

\section{Large-step Solving for Emerson-Lei Games}\label{sec:el}

In this section, we show how games with an Emerson-Lei objective
can be transformed to equivalent parity games, using an economic variant of the later-appearance-record (LAR) construction; we also show that this transformation
preserves dag sub-structures. Thereby we enable usage of the acceleration
method for parity games from the previous section also for Emerson-Lei games.

The LAR reduction annotates game nodes with permutations that act as a memory
and record the order in which colors have been recently visited; this allows
to identify the set of colors that is visited infinitely often. A parity
condition is used to check whether this set satisfies the Emerson-Lei objective.
A slight difference to previous LAR constructions for Emerson-Lei automata and games~\cite{RenkinDP20,HunterD05} is that our reduction moves at most one color within the LAR in each game step, thereby allowing to bound the branching of memory updates along DAG parts of games.

Before we present the formal reduction, we first fix a set $C$ of colors and introduce notation for permutations over $C$.
We let $\Pi(C)$ denote the set of permutations over $C$,
and for a permutation $\pi\in\Pi(C)$ and a position $1\leq i\leq |C|$, we let
$\pi(i)\in C$ denote the element at the $i$-th position of $\pi$.
Let $\pi_0$ be some fixed element of $\Pi(C)$.
For $D\subseteq C$ and $\pi\in\Pi(C)$, we let $\pi@D$ denote the permutation that is
obtained from $\pi$ by moving \emph{the} element of $D$ that occurs at the right-most position in 
$\pi$ to the front of $\pi$; for instance, for $C=\{a,b,c,d\}$ and $\pi=(a,d,c,b)\in\Pi(C)$,
we have $\pi@\{a,d\}=\pi@\{d\}=(d,a,c,b)$ and $(d,a,c,b)@\{a,d\}=\pi$.
Crucially, restricting the reordering to single colors, rather than sets of colors, ensures that for 
each $\pi\in \Pi(C)$, there are only $|C|$ many $\pi'$ such
that $\pi@D=\pi'$ for some $D\subseteq C$. 
Therefore, the branching in the reduced game is bounded in such a way
that a DAG with $o$ exit points and in which every path from the entry point to
some exit point visits at most one node with non-empty set of colors, is reduced to a
DAG that again has $o$ exit points.
 Given a permutation $\pi\in\Pi(C)$ and an index $1\leq i\leq |C|$, we furthermore let  $\pi[i]$ denote the set of colors that occur
in one of the first $i$ positions in $\pi$.

\begin{definition}\label{defn:lar}
Let $G=(V_\exists,V_\forall,E,\gamma_C,\varphi_C)$ be an Emerson-Lei game with set $C$ of colors.
We define the parity game
\begin{align*}
P(G)=(V_\exists\times\Pi(C),V_\forall\times\Pi(C),E',\Omega)
\end{align*}
by putting
$E'(v,\pi)=\{(w,\pi@\gamma_C(v))\mid (v,w)\in E\}$
for $(v,\pi)\in V\times\Pi(C)$, and
\begin{align*}
\Omega(v,\pi)=\begin{cases}
2p & \pi[p]\models\varphi_C\\
2p+1 & \pi[p]\not\models\varphi_C
\end{cases}
\end{align*}
for $(v,\pi)\in V\times\Pi(C)$. Here, $p$ denotes the right-most position in $\pi$
that contains some color from $\gamma_C(v)$.

\end{definition}

\begin{lemma}\label{lem:LAR}
For all $v\in V$, we have $v\in \mathsf{Win}_G$ if
and only if $(v,\pi_0)\in\mathsf{Win}_{P(G)}$.
\end{lemma}
\begin{proof}
For a finite play $\tau=(v_0,\pi_0)(v_1,\pi_1)\ldots (v_n,\pi_n)$
of $P(G)$,
we let $p(\tau)$ denote the 
right-most position in $\pi_0$ such that $\pi_0(p(\tau))\in\gamma_C(v_ 0v_1\ldots v_n)$. By slight abuse of notation, we let $\Omega(\tau)$ denote
the maximal $\Omega$-priority that is visited in $\tau$, having 
$\Omega(\tau)=2(p(\tau))$ if $\pi_0[p(\tau)]\models\varphi_C$
and
$\Omega(\tau)=2(p(\tau))+1$ if $\pi_0[p(\tau)]\not\models\varphi_C$.
\begin{itemize}
	\item[$\Rightarrow$] Let $\sigma=(\Pi(C),\mathsf{update}_\sigma,\mathsf{move}_\sigma)$ be a strategy 
with memory $\Pi(C)$ for player $\exists$ in $G$
with which they win every node from their winning region.
It has been shown in a previous LAR reduction for Emerson-Lei games~\cite{HunterD05} that winning strategies
with this amount of memory always exist.
We define a strategy $\rho=(\Pi(C),\mathsf{update}_\rho,\mathsf{move}_\rho)$
with memory $\Pi(C)$ for player $\exists$ in $P(G)$
by putting $\mathsf{update}_\rho(\pi,((v,\pi'),(w,\pi'')))=\mathsf{update}_\sigma(\pi,(v,w))$
and 
\begin{equation*}
\mathsf{move}_\rho((v,\pi'),\pi)=(\mathsf{move}_\rho(v,\pi),\pi@\gamma_C(v)).
\end{equation*} 
Thus $\rho$ updates the memory and picks moves just as 
$\sigma$ does, but also updates the permutation component in $P(G)$ according to the taken moves; hence $\rho$ is a valid strategy.

We show that $\rho$ wins a node $(v,\pi)$ in $P(G)$ whenever $v$ is in the winning region
of player $\exists$ in $G$.
To this end, let $\tau=(v_0,\pi_0)(v_1,\pi_1)\ldots$ be a play of $P(G)$ that starts at 
$(v_0,\pi_0)=(v,\pi)$ and is compatible with $\rho$. By construction, $\theta=v_0 v_1\ldots$ is a
play that is compatible with $\sigma$. Since $\sigma$ is a winning strategy for player
$\exists$, we have $\gamma_C(\theta)\models \varphi_C$. There is a number $i$ 
such that all colors that appear in $\theta$ from position $i$ on occur infinitely often.
Let $p$ be the number of colors that appear infinitely often in $\theta$.
It follows by definition of $\pi@D$ for $D\subseteq C$ (which moves the single
right-most element of $\pi$ that is contained in $D$ to the very front of $\pi$),
that there is a position $j\geq i$ such that the left-most $p$ elements of $\pi_j$
are exactly the colors occurring infinitely often in $\pi$ (and all colors to
the right of $\pi_j(p)$ are never visited from position $j$ on). It follows
that from position $j$ on, $\tau$ never visits a priority larger than $2p$.
To see that $\tau$ infinitely often visits priority $2p$
we note that $\pi_{j'}[p]\models\varphi_C$
for any $j'>j$, so it suffices to show that $p$ infinitely
often is the rightmost position in the permutation component of $\tau$ 
that is visited. This is the case since, from position $j$ on, 
the $p$-th element in the permutation component
of $\tau$ cycles fairly through all colors that are visited infinitely often by $\theta$.

	\item[$\Leftarrow$]
	
Let $\rho$ be a positional strategy for player $\exists$ in $P(G)$
with which they win every node from their winning region.
We define a strategy $\sigma=(\Pi(C),\mathsf{update}_\sigma,\mathsf{move}_\sigma)$
with memory $\Pi(C)$ for player $\exists$ in $G$
by putting $\mathsf{update}_\sigma(\pi,(v,w))=\pi@\gamma_C(v)$
and $\mathsf{move}_\sigma(v,\pi)=w$ where
$w$ is such that $\rho(v,\pi)=(w,\pi@\gamma_C(v))$. Thus $\sigma$ updates the memory and 
picks moves just as plays that follow $\rho$ do.

We show that $\sigma$ wins a node $v$ in $G$ whenever $(v,\pi)$ is in the winning region
of player $\exists$ in $P(G)$.
To this end, let $\theta=v_0 v_1\ldots$ be a play of $G$ that starts at 
$v_0=v$ and is compatible with $\sigma$. By construction, $\theta$
induces a play $\tau=(v_0,\pi_0) (v_1,\pi_1)\ldots$ with $(v_0,\pi_0)=(v,\pi)$
and $\pi_{i+1}=\pi_{i}@\gamma_C(v_i)$ for $i\geq 0$
that is compatible with $\rho$. Since $\rho$ is a winning strategy for player
$\exists$, the maximal priority in it is even (say $2p$).
Again, $p$ is the position such that the left-most $p$ elements in the permutation
component of $\rho$ from some point on contain exactly the colors that are
visited infinitely often by $\theta$. It follows that $\gamma_C(\theta)\models\varphi_C$.
\end{itemize}
\end{proof}
The LAR reduction from Definition~\ref{defn:lar} preserves DAG sub-structures:
\begin{lemma}
Let $G$ be an Emerson-Lei games with set $C$ of colors.
If $W$ is a positional DAG in $G$, then $W\times \Pi(C)$ is a positional
DAG in $P(G)$.
\end{lemma}
\begin{proof}
The claim follows immediately since the LAR reduction just annotates game nodes with additional memory, meaning that each game node $v$ in $G$ is replaced with copies $(v,\pi)$ in $P(G)$ so that all cycles in $P(G)$ have
corresponding cycles in $G$. 
\end{proof}

\section{Conclusion}

We have presented a method for solving parity games with DAG sub-structures
by computing nested fixpoints over the non-DAG nodes only, thereby intuitively accelerating fixpoint computation by summarizing cycle-free parts during the solution
of parity games. This can 
significantly reduce the domain of the game solving process, and in 
cases where DAG sub-structures in addition can be evaluated without exploring all
nodes in them, it improves the time complexity of parity game solving.
Furthermore, we have proposed a later-appearance-record construction with linear branching on the memory values
that transforms Emerson-Lei games to parity games, and have shown that this transformation preserves DAG sub-structures,
enabling usage of the proposed method also for the solution of games
with general Emerson-Lei objectives.

\bibliography{lib}

\end{document}